\documentclass[final,5p,times,twocolumn]{elsarticle}

\usepackage{amssymb}
\usepackage{braket}
\usepackage{amsmath}
\usepackage{amsfonts}
\usepackage{amsthm}
\usepackage{amssymb}
\usepackage{tikz}
\usepackage{yhmath}
\usepackage{multirow, hyperref}
\usepackage{makecell, ulem}
\usepackage{url}
\usepackage[ruled,linesnumbered]{algorithm2e}
\usetikzlibrary{arrows}

\newtheorem{theorem}{Theorem}
\newtheorem{prop}{Proposition}
\newtheorem{lemma}{Lemma}

\newtheorem{corollary}{Corollary}

\newtheorem{definition}{Definition}

\newcommand \mc {\text{ Max-Cut}}

\journal{...}

\begin{document}

\begin{frontmatter}



\title{Comparison of Hyperplane Rounding for Max-Cut and Quantum Approximate Optimization Algorithm over Certain Regular Graph Families}


\author[GT]{Reuben Tate}
\affiliation[GT]{organization={CCS-3 Information Sciences, Los Alamos National Lab}, country={USA}}
\author[MIT]{Swati Gupta\corref{cor1}}
\affiliation[MIT]{organization={Sloan School of Management, MIT}, country={USA}}
\cortext[cor1]{Corresponding Author. Email: swatig@mit.edu}

\begin{abstract}
There is a strong interest in finding challenging instances of NP-hard problems, from the perspective of showing quantum advantage. Due to the limits of near-term NISQ devices, it is moreover useful if these instances are small. In this work, we identify two graph families ($|V|<1000$) on which the Goemans-Williamson algorithm for approximating the Max-Cut achieves at most a 0.912-approximation. We further show that, in comparison, a recent quantum algorithm, Quantum Approximate Optimization Algorithm (depth $p=1$), is a 0.592-approximation on Karloff instances in the limit ($n \to \infty$), and is at best a $0.894$-approximation on a family of strongly-regular graphs. We further explore construction of challenging instances computationally by perturbing edge weights, which may be of independent interest, and include these in the \href{https://github.com/swati1729/CI-QuBe}{CI-Qube} github repository.

\end{abstract}




\end{frontmatter}


\section{Introduction}
Demonstrating quantum advantage\footnote{We note that there is not universal agreement regarding the precise definition of quantum advantage. Many have used the term \emph{quantum supremacy} to refer to what we call quantum advantage in this work; meanwhile, many will instead use the term quantum advantage to mean a situation where, for some given task, a physical quantum devices has at least some \emph{slight} edge in runtime compared to any classical algorithm on any classical machine. Regardless of the definition of quantum advantage that one uses, we believe that the instances presented in this work should be of interest to the reader.}, i.e. the ability for physical quantum computers to complete a given task that no classical computer can complete in a reasonable amount of time, has been of particular interest to quantum researchers. In 2019, Google had announced that they had demonstrated quantum advantage on a 53 qubit quantum device; in particular, they claimed that their quantum device could sample from random quantum circuits that would take ``about 10000 years to complete"  \cite{Google2019}. While classical devices' ability to simulate quantum systems is of great interest, it remains an open question which scientifically and commercially relevant problems could exhibit a quantum advantage. 
%\sg{we shouldn't downplay google.} \rt{I rewrote this a bit to be more neutral. I should note that people have pushed against Google's claims and there are papers in 2022 that show that classical computers can match or even beat Google's Sycamore processor on the same problem.} 
On the other hand, there exists algorithms, such as Shor's algorithm for integer factorization, which are theoretically faster than any currently known classical algorithm; however, to be of practical use, Shor's algorithm requires a large number of qubits which is currently not feasible for current quantum devices \cite{S94}.

Many believe that the Quantum Approximate Optimization Algorithm (QAOA) applied to the Max-Cut problem has the potential for demonstrating a quantum advantage over classical algorithms \cite{HTOLHS21, ZWCPL20, bayerstadler2021industry, campbell2022qaoa, FH16, shaydulin2021classical}. This is in part due to a recent result by Farhi et al. \cite{FGG14} which shows that under the adiabatic limit (i.e., as the circuit depth goes to infinity), the QAOA algorithm would converge to the Max-Cut. It may be the case though that the QAOA (and its variants) only have an advantage over classical algorithms on certain classes of graphs for finite circuit depths. To identify such classes of graphs, one can take a two-pronged approach: (i) find graph families which are challenging for classical optimization algorithms, (ii) show that the QAOA performs well on such graph families. From the perspective of showing quantum advantage for certain families of graphs on near-term NISQ devices, it is further of interest to find such relatively small graphs. 

In this work, we further our knowledge of challenging graph instances for Max-Cut, where classical algorithms may not be able to achieve the optimal solutions easily. In 2018, Dunning et al. \cite{dunning2018} showed that for graphs upto around ~750 nodes, Max-Cut can be solved optimally (in reasonable time) and verified using the suite of heuristics implemented in their library \texttt{MQLib}. We therefore turn our attention to the Goemans-Williamson (GW) algorithm \cite{GW95}, as a competing benchmark\footnote{Many other classical algorithms exist for the Max-Cut problem including Trevisan's algorithm, which yields a $0.614$-approximation \cite{Trevisan12, Soto15}, and numerous heuristics such as those found in the MQLib library \cite{dunning2018}. The task of identifying classes of instances that demonstrate quantum advantage over \emph{all} classical algorithms is difficult; therefore, we consider the tractable task of demonstrating an advantage over the best-known classical algorithm of Goemans \& Williamson \cite{GW95}. Additionally, we do not compare to higher-degree relaxations of Max-Cut due to their prohibitively higher computational cost \cite{lasserre2001global,parrilo2000structured}.}, %remark that in the context of Sum of Squares (SOS) hierarchies, the GW relaxation is equivalent to a degree-2 SOS relaxation; higher-degree relaxations can potentially be used, thus obtaining stronger relaxations, but this requires a 
which enjoys provable and best-possible approximation ratio of 0.878 in polynomial time for graphs with non-negative edge-weights. Any classical algorithm cannot attain a provable approximation of better than 0.878, assuming that the Unique Games Conjecture is true \cite{Khot02,MOO05,KKMO07}. 

%\sout{In this work, we identify such classes of instances which have the potential for demonstrating some form of quantum advantage. These classes of instances contain several graphs with fewer than 1000 nodes, making them suitable for current and near-term quantum devices; we have made these instances available online in the Combinatorial Instances for Quantum Benchmarking (CI-QuBe) library \cite{CI-QuBe2021}. As the name of our library suggests, these instances may also be of interest from a quantum bencharmking perspective as well.}

In 1996, Karloff \cite{K99} demonstrated a sequence of instances where the instance-specific approximation ratio achieved by the GW algorithm approaches the 0.878 approximation guarantee. These constructed instances are known to be highly symmetric and in particular, are vertex transitive.\footnote{A graph $G = (V,E)$ is considered to be vertex-transitive if, for any $u,v \in V$ there exists a graph automorphism $f:V \to V$ such that $f(u) = v$.} In Section \ref{sec:GW_Alg}, we use Karloff's construction and a proof of one of his conjectures \cite{BCIM18} to enumerate a collection of classically challenging instances that are suitably-sized for current and near-term quantum devices.

%\sge{TODO: First explain the search for symmetric instances, since there exist families that are worst-case for GW, yet, Herrman et al. show QAOA runs better on symmetric graphs. }\rt{Done.}

In Section \ref{sec:stronglyRegularApprox}, we prove that the GW algorithm also achieves a low instance-specific approximation ratio for a special class of strongly-regular graphs. In particular, for strongly-regular graphs parameterized by $(n,k,\lambda, \mu)$ with $n=4(3t+1), k=3(t+1), \lambda=2, \mu=t+1$ for some integer $t$, we prove (Theorem \ref{thm:q3tAR}) that the GW algorithm yields a 0.912-approximation whenever $\text{Max-Cut}(G) = \frac{2}{3}|E|$ (which is the case for nearly all graphs in this family). In general, strongly-regular graphs are not necessarily as symmetric as the instances constructed by Karloff \cite{K99}; however, they do have a local symmetry which quantum researchers have been able to exploit in the context of quantum walks \cite{janmark2014global}.

In regards to symmetry, Herrman et al. \cite{HTOLHS21} ran numerical simulations for Max-Cut QAOA, up to depth $p=3$, for all non-isomorphic graphs with 8 or fewer nodes, in order to find correlations between various graph properties and performance metrics (including the approximation ratio and the probability of observing the maximum cut); in particular, they found that QAOA performs empirically better on graphs with higher amounts of symmetry, thus further motivating the investigation of the classes of instances considered in this work.

%\sg{Cite Herrman 2021 in the related wokr.}
%\sout{In this work, we specifically focus our attention to classes of instances for which the Goemans-Williamson (GW) algorithm achieves instance-specific approximation ratios that are bounded considerably away from 1. \rte{Should we cite the thesis here and the fact that heuristics do really well, thus motivating us to instead consider the GW algorithm?}}

%In this work, we present three classes of instances which may be of potential interest in regards to quantum advantage or benchmarking. For the first class, 

%In general, for nearly all instances, we find that there is typically at least one heuristic that obtains a nearly optimal cut, thus, demonstrating quantum advantage over \emph{every} classical algorithm may prove to be an insurmountable challenge. Instead, for the second class of graphs, we consider instances for which the best-known classical approximation algorithm with approximation ratio $0.878$, i.e. the Goemans-Williamson algorithm, performs poorly; in particular we consider a sequence of graphs proposed by Karloff \cite{K99} whose GW instance-specific approximation ratio approaches the GW bound of 0.878. As these instances grow quickly in size (making them not suitable for near-term quantum devices), we also consider perturbations of the smaller of these instances.

In Section \ref{sec:qaoaPerformanceInterestingInstances}, we discuss QAOA's performance on the instances in Section \ref{sec:GW_Alg} where the GW algorithm provably performs poorly. We prove (Theorem \ref{thm:karloffQAOAApproxRatio}) that depth-1 QAOA achieves an instance-specific approximation ratio that approaches 0.592 for the sequence of graphs constructed by Karloff; thus, a higher-depth or some modification of the QAOA algorithm (e.g., using warm-starts \cite{TFHMG20,TMGMG21,egger2020warm}) is required in order to demonstrate some form of quantum advantage for these instances. As for the family of strongly-regular graphs defined above, we computationally verified that depth-1 QAOA achieves instance-specific approximation ratios in $[0.82, 0.90]$ for each instance.

%\rte{Finally, we provide a discussion in Section \ref{sec:discussion} where we discuss the importance of such instances for the purposes of benchmarking and also provide the reader with additional information regarding the Ci-QuBe library \cite{CI-QuBe2021} which contains all of the instances discussed in this paper and more.}

We close the paper with a discussion on the performance of classical and quantum algorithms on perturbed instances generated using strongly-regular graph families and Karloff graphs (included in a github repository \href{https://github.com/swati1729/CI-QuBe}{CI-Qube}). We conjecture that the mixed-weight graphs form an important class for benchmarking, using preliminary evidence from \cite{dunning2018}.

%\sg{What about other algorithms? Higher heirarchies, or a different methods for solving Max-Cut on these instances? Say something concluding about these.}

%For strongly-regular graphs, the proof of Theorem \ref{thm:karloffQAOAApproxRatio} can not directly be extended as our proof uses a result by Wang et al. \cite{WHJR18} that is only applicable for triangle-free graphs and strongly-regular graphs are not triangle-free in general.

%Finally, similar to the second class of instances, the last class also identifies instances for which GW performs poorly. The unit-weight instances in this third class, some of which are found in the MQLib library, have certain properties of interest: the maximum cut contains two thirds of the total number of edges in the graphs, the angles between vertices in the optimal GW SDP solution correspond to angles found in a regular tetrahedron, and, as a result of these properties, GW achieves an instance-specific approximation ratio of 0.912. We prove that a certain class of graphs $\mathcal{Q}$, namely, pseudo-geometric (strongly-regular) graphs for generalized quadrangles of order $(3,t)$ (with $t \in \{1,3,5,9\}$) satisfy these properties. As $\mathcal{Q}$ is finite, this proof is done via an exhaustive computer verification; however, we provide a partial proof (not using computer verification) which has the potential to be generalized for classes of instances similar to this third class.
\section{Preliminaries}
\label{sec:preliminaries}
Given a graph $G=(V,E)$ with weights $w:E \to \mathbb{R}$ and $|V|=n$, the goal of the weighted Max-Cut problem is to partition the vertices into two parts so that the sum of weights of edges across the parts is maximized; letting each vertex correspond to a $\{-1,+1\}$ decision variable corresponding to which side of the cut it is on, the Max-Cut problem can be formulated as $$\text{Max-Cut}(G) = \max_{\substack{x \in \mathbb{R}^{n}:\\ x_i \in \{-1,1\}}} \frac{1}{2}\sum_{(i,j) \in E} w_{ij}(1-x_i\cdot x_j).$$ In the seminal work by Goemans and Williamson \cite{GW95}, they relax this formulation to  \begin{equation}\max_{\substack{x \in \mathbb{R}^{n}:\\ x_i \in \mathbb{R}^n \\ \Vert x_i \Vert  = 1}} \frac{1}{2}\sum_{(i,j) \in E} w_{ij}(1-x_i\cdot x_j), \tag{VP}\end{equation} i.e., each $x_i$ is now a unit-length vector in $\mathbb{R}^n$. The vector-program (VP) above can be reformulated as the following semidefinite program (SDP):
\begin{equation}
    \begin{array}{lll@{}ll}
		&\text{maximize}  &  & \frac{1}{2}\sum_{(i,j) \in E} w_{ij}(1-Y_{ij}) &\\
		&\text{subject to}&
		& Y_{ii} = 1,  & \forall i \in [n],\\
		& &   & Y \in \mathbb{S}^n_+\,, &
	\end{array} \tag{SDP}
\end{equation}
where $\mathbb{S}^n_+$ is the set of all $n \times n$ positive semidefinite matrices \cite{GW95}. Given an optimal solution $Y$ to (SDP), one can obtain an optimal solution to (VP) by considering the Cholesky factorization $Y = x^Tx$ with $x$ being an $n \times n$ matrix and letting $x_i$ be the $i$th column of $x$ for each $i \in [n]$. Goemans and Williamson show that hyperplane rounding of $x$ (i.e. selecting a random hyperplane through the origin of $\mathbb{R}^n$ and partitioning the vertices into two sets according to which side of the hyperplane they lie on) yields an expected cut value of
$$\text{HP(x)} = \sum_{(i,j) \in E} \frac{w_{i,j}}{\pi} \arccos(x_i \cdot x_j),$$
and yields a worst-case approximation ratio of $0.878$ for the GW algorithm \cite{GW95}.

For graphs $G$ with non-negative edge weights, we define the instance-specific approximation ratio for the GW algorithm as $\alpha_{GW,G} = \frac{\text{HP}(x)}{\text{Max-Cut}(G)}$ where $Y=x^Tx$ is an optimal solution to (SDP).

\subsection{Quantum Approximate Optimization Algorithm}
Farhi et al. \cite{FGG14} proposed the Quantum Approximate Optimization Algorithm (QAOA) as a general framework for obtaining approximate solutions to certain classes of combinatorial optimization problems, including Max-Cut. The QAOA algorithm is a quantum algorithm parametrized by a circuit depth $p \in \mathbb{Z}^{+}$ and variational parameters $\gamma,\beta \in \mathbb{R}^p$ which are tuned by a classical computer. We let $F_p(\gamma,\beta)$ denote the expected cut value of depth-p QAOA at a particular choice of $\gamma$ and $\beta$, and $F_p(G) = \max_{\gamma,\beta \in \mathbb{R}^p} F_p(\gamma,\beta)$ denote the expected cut value of depth-$p$ QAOA at optimal choice of $\gamma$ and $\beta$. QAOA performance in monotonic in the sense that if $q \geq p$, then $F_{q}(G) \geq F_p(G)$ \cite{FGG14}. We direct the reader to \cite{FGG14} for more details about QAOA.

\section{Small Instances from Karloff's Construction with Low GW Approximation Ratios}
\label{sec:GW_Alg}
We review a known graph family constructed by Karloff in \cite{K99} where the GW algorithm attains the tight approximation ratio of 0.878 as the graph size increases. The smallest instance using Karloff's construction has 924 nodes, which is too large for most current quantum devices. Although it is possible to classically calculate QAOA expectations for large instances by exploiting the locality of QAOA \cite{augustino2024strategies}, such strategies become computationally infeasible as the degrees of the vertices and the number of QAOA layers $p$ increases; in particular, the study by \cite{augustino2024strategies} considers 3-regular graphs whereas the smallest instance using Karloff's construction is a 225-regular graph.

%{\color{red} not true? Farhi's collab for 3 degree graphs with 1024 nodes?} \rt{I've tried to address this better, see edits. I've also edited the table of Karloff instances to show the degrees.}  
Instead, in this work, we studied the construction to increase the size of the graph family by showing properties of the eigenvalues of the adjacency matrix of the graphs in Karloff's construction \cite{BCIM18}. This allowed us to expand the set of graphs where the GW algorithm has low approximation ratio to include graphs of 20 nodes and above (see Table \ref{tab:karloff}).

\subsection{Review of Karloff's Construction}
Goemans and Williamson \cite{GW95} showed that for a graph $G=(V,E)$ with non-negative edge weights,
$$\alpha_{\text{GW},G} \ge \alpha^*, \text{ where } \alpha^* := \frac{2}{\pi}\min_{\theta \in [0, \pi]} \frac{\theta}{ 1 - \cos \theta};$$  calculating $\alpha$ one finds that $0.878 < \alpha^* < 0.879$.

However, the typical analysis of Goemans-Williamson does not show that this is tight; that is, could there perhaps be a constant $\beta^* > \alpha^*$ such that $\alpha_{\text{GW},G} \geq \beta^*$ for all graphs $G$? Karloff \cite{K99} showed that the GW algorithm can indeed guarantee no better than an $\alpha^* \approx 0.878$ approximate solution, by constructing simple graphs for which the expected performance of the algorithm is arbitrarily close to $\alpha^*$.

These challenging instances for the Goemans-Williamson algorithm are explained next. For non-negative integers $b \leq t \leq m$, let $J(m,t,b)$ denote the graph with vertex set ${[m] \choose t}$, i.e., the vertices are all $t$-element subsets of $[m]$; two distinct vertices/subsets $S$ and $T$ of $J(m,t,b)$ are adjacent if and only if they have exactly $b$ elements in common, i.e. $|S \cap T| = b$. Of particular interest are graphs $J(m,t,b)$, with $m$ even, $t = m/2$.

By choosing $m,t,b$ appropriately as seen in Theorem \ref{Thm-Karloff-exists}, Karloff shows that this construction yields instances that are arbitrarily close to the $\alpha^* \approx 0.878$ bound.

\begin{theorem}[Karloff \ \cite{K99}]\label{Thm-Karloff-exists} There exists an optimal solution $Y$ of the GW SDP relaxation such that for each $\varepsilon > 0$, there are $b$ and $m$, $m$ even and positive and $0 \leq b \leq m/12$, such that $\alpha_{\text{GW},G} \leq \alpha^*+\varepsilon$ where $G = J(m,m/2,b)$.
\end{theorem}

The construction of the matrix $Y$ in Theorem \ref{Thm-Karloff-exists}, given by Karloff \cite{K99}, is as follows. For each vertex/subset $S$ of $J(m,m/2,b)$, the vector $w_S \in \mathbb{R}^m$ is constructed where the $i$th entry of $w_S$ is $+1$ if $i \in S$ and $-1$ if $i \notin S$. Then, each $w_S$ is rescaled by $1/\sqrt{m}$ so that the $w_S$'s are of unit length. Now, an $m \times n$ matrix $W$ is built by letting the column indexed by $S$ be the vector $x_S$. Finally, $Y$ is set to $Y = W^TW$. The feasibility of $Y$ (with respect to the GW relaxation) can be easily verified. Karloff shows that $Y$ is also an optimal solution provided that $0 \leq b \leq m/12$ (a condition in Theorem \ref{Thm-Karloff-exists}).

%Note that there may be more than one optimal solution to the GW relaxation; however, Karloff only shows the ratio above is achieved only for the specific solution he constructs. Also note that the ratio is in terms of the \emph{expected} value of the cut; it may be the case that there are specific hyperplanes that give cuts whose values are much better than the expected value.

In order to better understand the instances that satisfy the conditions of Theorem \ref{Thm-Karloff-exists} and the GW algorithm's performance on such instances, we next review the specific approximation ratio that is achieved as a function of the parameters $m$ and $b$ used in Karloff's construction \cite{K99}.

\begin{theorem}[Karloff \ \cite{K99}]\label{Thm-Karloff} Let $m$ be an even positive integer and $G = J(m,m/2,b)$. If $0 \leq b \leq m/12$, then, using the optimal solution $\hat{Y}_G$,
\begin{align*}
&\alpha_{\text{GW},G} = \frac{\frac{1}{\pi}\arccos(\frac{4b}{m}-1)}{1-\frac{2b}{m}} = \frac{2}{\pi}\frac{\theta}{1-\cos(\theta)}\\
&\geq \min_{\theta' \in [0,\pi]} \frac{2}{\pi}\frac{\theta'}{1-\cos(\theta')} = \alpha^*, \text{ where } \theta = \arccos\left(\frac{4b}{m}-1\right).\end{align*}
\end{theorem}

Numerically calculating the minimizer $\theta^*$ in the inequality in Theorem \ref{Thm-Karloff} yields $\theta^* \approx 2.33112$. Thus, the closer $\theta = \arccos(4b/m-1)$ is to $\theta^* \approx 2.33112$, the lower the instance-specific approximation ratio. Equivalently, in order to minimize the instance-specific approximation ratio, the ratio $b/m$ should be chosen to be as close to $\frac{1}{4}(\cos(\theta^*)+1) \approx 0.0777$ as possible; note that this can be achieved by picking $m$ large enough and then picking a suitable $b$. Furthermore, if one picks $b/m \approx 0.0777$, then the conditions of Theorem \ref{Thm-Karloff} still hold as $b \approx 0.0777m \leq m/12$. Since it is clear that $b$ and $m$ can be chosen to make the approximation ratio arbitrarily close to $\alpha^*$, Theorem \ref{Thm-Karloff-exists} follows.

%The number of qubits that can be used on near-term quantum devices is extremely limited; thus, amongst the instances generated by Karloff's construction, it is important to identify the instances which are small enough to run on such devices. We begin with a theorem that describes the formula for the instance-specific approximation ratio of instances made using Karloff's construction.
\subsection{Identification of Small Instances Using Karloff's Approach}
\label{sec:small_instances_karloff}
Theorem \ref{Thm-Karloff} seems to yield a promising approach for finding instances $G$ where $\alpha_{\text{GW},G}$ is small; however, Theorem \ref{Thm-Karloff} is not sufficient for finding small instances that are feasible for near-term quantum computers. To illustrate why this is the case, we first consider the case where $b = 0$. In this case, $\theta = \arccos(-1) = \pi$ and thus $\alpha_{\text{GW},J(m,m/2,0)}$ $ = \frac{2}{\pi} \frac{\pi}{1-\cos(\pi)} = 1$ which is not very interesting. If $b > 0$, then the condition $b \leq m/12$ in Theorem \ref{Thm-Karloff} implies that $m \geq 12$ in which case the graph $J(m,m/2,b)$ would have $n \geq {12 \choose 6} = 924$ nodes which is much more than the number of qubits in most modern quantum computers.

One way to obtain smaller interesting instances is to find a way to relax the condition $b \leq m/12$ in Theorem \ref{Thm-Karloff}. This condition is needed since Theorem \ref{Thm-Karloff} invokes the following theorem.

\begin{theorem}[Karloff \ \cite{K99}]\label{Thm-Karloff-2}
\label{thm:smallestEigenvalue}
 Let $m$ be an even positive integer and let $0\leq b \leq m/12$. The smallest eigenvalue of the adjacency matrix of $J(m,m/2,b)$ is
$ {m/2 \choose b}^2 \left[ \frac{4b}{m}-1 \right].$
\end{theorem}

Karloff used the above theorem to construct a solution to the dual SDP that is feasible and has the same objective value as the construction of the primal $Y$ described earlier, thus proving the optimality of $Y$; we utilize a similar technique later in Section \ref{sec:stronglyRegularApprox} on strongly-regular graphs. The proof of feasibility of the dual is the only place where the condition $0 \leq b \leq m/12$ is used in Karloff's proofs; meaning that if the condition can be relaxed in Theorem \ref{thm:smallestEigenvalue} above, that would expand the set of instances for which Theorem \ref{Thm-Karloff} holds.

Karloff conjectured that the condition $0 \leq b \leq m/12$ in Theorem \ref{thm:smallestEigenvalue} could be relaxed to the condition $0 \leq b < m/4$. We found that this conjecture was in fact proven by Brouwer et al. in 2018 (see remark after Theorem 3.10 in \cite{BCIM18}), therefore the same condition in Theorem \ref{Thm-Karloff} can also be relaxed;
 we restate Theorem \ref{Thm-Karloff} with the relaxed inequality below.

\begin{theorem}[Karloff \ \cite{K99}]\label{Thm-Conj} Let $m$ be an even positive integer and $G = J(m,m/2,b)$. If $0 \leq b < m/4$, then 
\begin{align*}
&\alpha_{\text{GW},G} = \frac{\frac{1}{\pi}\arccos(\frac{4b}{m}-1)}{1-\frac{2b}{m}} = \frac{2}{\pi}\frac{\theta}{1-\cos(\theta)}\\
&\geq \min_{\theta' \in [0,\pi]} \frac{2}{\pi}\frac{\theta'}{1-\cos(\theta')} = \alpha^*, \text{ where } \theta = \arccos\left(\frac{4b}{m}-1\right).
\end{align*}

\end{theorem}

From Theorem \ref{Thm-Conj}, we generate instances $J(m,m/2,b)$ where $m$ is as small as $m=6$ and whose GW instance-specific approximation ratios are at most 0.940; these instances and approximation ratios are displayed in Table \ref{tab:karloff}. These smaller instances are suitable for near-term NISQ devices.

\begin{table}
\centering
\begin{tabular}{|c|c|c|c|c|c|}
\hline
     Instance $G$ &  Nodes & Edges & Deg. & $\alpha_{\text{GW},G}$ & $\alpha_\text{QAOA}$\\\hline
     $J(6,3,1)$ & 20 & 90 & 9 & $0.9123$ & $0.8492$ \\
     $J(8,4,1)$ & 70 & 560 & 16 & $0.8889$ & $0.7694$\\
     $J(10,5,1)$ & 252 & 3150 & 25 & $0.8810$ & $0.7016$ \\
     $J(10,5,2)$ & 252 & 12600 & 100 & $0.9402$ & $0.8526$  \\
    $J(12,6,1)$ & 924 & 16632 & 36 & $0.8787$ & $0.6611$ \\
    $J(12,6,2)$ & 924 &103950 & 225 & $0.9123$ & $0.7654$ \\
    \hline
\end{tabular}
\caption{\label{tab:karloff} A listing of small ($<1000$ nodes) instances using Karloff's construction. For each instance, we include the number of nodes, edges, the degree (of each node), and the theoretical instance-specific approximation ratios one would obtain if running the single-layer QAOA algorithm (see Section \ref{sec:qaoaPerformanceInterestingInstances}) and the GW algorithm on that instance (assuming an optimal solution of $Y$ as described earlier).}
\end{table}

\section{Provable Guarantees for the GW Algorithm on Strongly-Regular Graphs}
\label{sec:stronglyRegularApprox}
In Karloff's proof of Theorem \ref{Thm-Karloff}, he exploits the fact that, for any  fixed instance $J(m,m/2,b)$ with $0\leq b \leq m/12$, angles between adjacent vertices in the optimal GW SDP solution are all equal (with angle $\theta = \arccos\left( \frac{4b}{m}-1\right)$). This raises the question: \textit{are there other graphs whose optimal GW SDP solution has equal angles between adjacent vertices, and if so, what is the GW approximation ratio for such graphs}? This subsection answers this question in the affirmative: we prove that such a property holds for all strongly-regular graphs. Additionally, for a specific family of strongly-regular graphs (i.e. those parameterized by $n=4(3t+1), k=3(t+1), \lambda=2, \mu=t+1$ for some integer $t$) we show that the GW algorithm yields an instance-specific approximation ratio of 0.912 for nearly all instances in this family (Theorem \ref{thm:q3tAR}).

\subsection{General SRGs}
We first focus on general SRGs. First, in Proposition \ref{thm:equalAngles}, we generalize a portion of Karloff's proof to show that hyperplane rounding of the GW algorithm yields a cut with $\frac{\theta}{\pi}|E|$ edges in expectation for any unit-weight graph whose optimal GW SDP solution has equal angles $\theta$ between adjacent vertices. Then, we define the notion of strongly-regular graphs (Definition \ref{def:srg}) and exploit the well-known fact that the vertices of such graphs can be mapped onto a unit-hypersphere so that the angles between adjacent vertices are all equal (Theorem \ref{thm:srgToSphere}) \cite{seidel1979strongly}. We then prove (Theorem \ref{thm:optimality}) that this mapping corresponds to an optimal GW SDP solution by finding a feasible solution to the dual SDP with the same objective value. 
% \sout{Finally, we prove (Proposition \ref{thm:srgMaxcut}) that for this family of graphs, the maximum cut contains exactly two-thirds of the edges; this last step then allows us to conclude in Theorem \ref{thm:q3tAR} that the GW algorithm yields a 0.912 instance-specific approximation ratio for instances in this family.} 
% \rte{Finally, we computationally verify that for all but 13 instances in this family of graphs, the maximum cut contains exactly two-thirds of the edges; this last step then allows us to conclude in Theorem \ref{thm:q3tAR} that, with the exception of the 13 instances, the GW algorithm yields a 0.912 instance-specific approximation ratio for instances in this family.}

We begin with Proposition \ref{thm:equalAngles} which relates the angles in the SDP solution to the GW approximation ratio in the case that the graph has unit-weight edges and equal angles between all adjacent vertices in the optimal GW SDP solution.

\begin{prop}
\label{thm:equalAngles}
If $G=(V,E)$ is a unit-weight graph with optimal GW SDP solution $Y=x^Tx$ and if there exists $\theta \in \mathbb{R}$ such that $\arccos(x_u \cdot x_v) = \theta$ for all $(u,v) \in E$, then the GW algorithm (with hyperplane rounding on $x$) yields exactly $\frac{\theta}{\pi}|E|$ edges in expectation and obtains an instance-specific approximation ratio of at least $\frac{\theta}{\pi}$.
\end{prop}
\begin{proof}
    Let $G=(V,E)$ be a unit-weight graph with optimal GW SDP solution $Y=x^Tx$ such that, for some $\theta \in \mathbb{R}$, $\arccos(x_u \cdot x_v) = \theta$ for all $(u,v) \in E$. Recall that $\text{HP}(x)$ denotes the expected number of edges obtained from applying hyperplane rounding to $x$. Thus, \begin{align*}
    \text{HP}(x) =& \sum_{(i,j) \in E} \frac{1}{\pi}\arccos(x_i \cdot x_j) =  \sum_{(i,j) \in E} \frac{1}{\pi}\theta = \frac{\theta}{\pi}|E|.
\end{align*}

Note that since $\text{\mc}(G) \leq |E|$ for any unweighted graph, thus the approximation ratio of hyperplane round is 
    $\geq \frac{\theta}{\pi}$. \qedhere
\end{proof}

In 1963, Bose introduced the notion of strongly-regular graphs \cite{bose1963strongly}, parameterized by\footnote{In the literature for strongly-regular graphs, one typically uses the parameter $v$ instead of $n$ to represent the number of vertices; this was done to avoid confusion with the other notation used throughout this work. It should also be noted that for any fixed set of parameters $(n,k,\lambda,\mu)$, there may be more than one non-isomorphic strongly-regular graph with those parameters.} $n,k,\lambda,$ and $\mu$. These graphs are highly symmetric. From this symmetry, as seen by Theorem \ref{thm:srgToSphere}, Seidel \cite{seidel1979strongly} shows that the vertices of an $n$-node strongly-regular graphs can be mapped to an $(r-1)$-dimensional hypersphere with $r \leq n$ so that, points corresponding to adjacent vertices all have the same instance-dependent angle $\theta$.
\begin{definition}
\label{def:srg}
A unit-weight graph $G$ is a $\text{SRG}(n,k,\lambda,\mu)$, i.e., $G$ is a strongly-regular graph with parameters $n,k,\lambda,\mu$, if $G$  has $n$ nodes, $G$ is a $k$-regular graph, every pair of adjacent vertices in $G$ have $\lambda$ common neighbors, every pair of non-adjacent vertices in $G$ have $\mu$ common neighbors, and $G$ is neither a complete graph nor the complement of a complete graph. A strongly-regular graph $G$ is said to be \emph{primitive} if both $G$ and its complement are connected.
\end{definition}

\begin{theorem}[Seidel \ \cite{seidel1979strongly}]
\label{thm:srgToSphere}
    Let $G=(V,E)$ be a primitive strongly regular graph $\text{SRG}(n,k,\lambda,\mu)$. Let $\xi_2$ be the smallest eigenvalue of the adjacency matrix of $G$. Then there exists an integer $r\leq n$ and a function $f: V \to \mathbb{S}^{r-1}$ so that for all $(u,v) \in E$, $f(u) \cdot f(v) = \xi_2/k$.
\end{theorem}

If one performs hyperplane rounding on the vectors obtained from Theorem \ref{thm:srgToSphere}, then the angles between vertices would be $\theta = \arccos(\xi_2/k)$, and following the proof of Theorem \ref{thm:equalAngles}, one would obtain an approximation ratio of at least $\theta / \pi = \frac{1}{\pi}\arccos(\xi_2/k)$. However, the performance of the GW algorithm is dependent on the \emph{optimal} solution to the SDP relaxation. Since the relationship between the optimal SDP solution and the equal-angular mapping above is not known a priori, it is not immediately clear if Theorem \ref{thm:srgToSphere} can be used to produce a non-trivial bound on the performance of GW on strongly-regular graphs; however, it can be shown that this equal-angular solution can be lifted to a higher dimension in which it is indeed an optimal SDP solution.
%In order to demonstrate such a proof, we first utilize a closed-form analytical expression for the eigenvalues of the adjacency matrices of strongly-regular graphs (Theorem \ref{thm:srgEigenvalues}) to the family above in Proposition \ref{thm:eigenForQ3t}; this allows us to determine (in Corollary \ref{thm:angleForQ3t}) the angles between adjacent vertices in the spherical mapping used in Theorem \ref{thm:srgToSphere} for the specific family of strongly-regular graphs mentioned above.
%In order to apply Proposition \ref{thm:equalAngles} to the family of graphs and the SDP solution corresponding to the mapping $f$ obtained from Corollary \ref{thm:angleForQ3t} above, we must show that such a solution is optimal (with respect to the GW SDP). 
We prove this optimality using a standard technique in optimization: find a feasible solution to the dual optimization problem with the same objective value. 
%In particular, we show (Propositions \ref{thm:feasiblePrimalSRG} and \ref{thm:feasibleDualSRG}) that the primal GW SDP and its dual have an optimal objective value of $\frac{2}{3}|E|$.

% \rt{Before proving optimality for this specific collection of SRGs, we first provide a few more helpful propositions that apply to general SRGs.}

For ease of notation, for feasible solutions $Y$ of the primal GW SDP, we let $z_P(Y)$ denote the objective value of the primal GW SDP at $Y$. Similarly, for feasible solutions $\zeta$ of the dual to the GW SDP, we let $z_D(\zeta)$ denote the objective value of the dual at $\zeta$. Additionally, we let $z_P^* = \max_{Y} z_P(Y)$ and $z_D^* = \min_{\zeta} z_D(\zeta)$ represent the optimal objective values of the GW SDP and its dual, where the max and min are taken over all feasible solutions of the SDP and its dual (respectively).

\begin{prop}
\label{thm:feasiblePrimalSRG}
Let $G$ be a $\text{SRG}(n,k,\lambda,\mu)$.  Let $\xi_2$ be the smallest eigenvalue of the adjacency matrix of $G$. Then there exists a feasible $Y$ for the GW SDP so that $z_P(Y) = \frac{|E|}{2}(1-\xi_2/k) $.
\end{prop}
\begin{proof}
    Let $G = (V,E)$ be such a graph. By Theorem \ref{thm:srgToSphere}, there exists an integer $r \leq n$ and a function $f: V \to \mathbb{S}^{r-1}$ so that for all $(u,v) \in E$, we have $f(u) \cdot f(v) = \xi_2/k$. We can embed the solution onto a sphere that lies in $\mathbb{R}^n$ by extending the function $f: V \to \mathbb{S}^{r-1}$ to the function $\hat{f}: V \to \mathbb{S}^{n-1}$ where, for all $i,j \in [n]$,
$$\hat{f}(i)_j = \begin{cases} f(i)_j, & 1 \leq j \leq r\\ 0, & r+1 \leq j \leq n,\end{cases}$$
where $\hat{f}(i)_j$ denotes the $j$th entry of $\hat{f}(i)$. It is straightforward to see that angles are preserved, i.e., for all $u,v \in V$, $\hat{f}(u) \cdot \hat{f}(v) = f(u) \cdot f(v)$.

Let $x$ be an  matrix where the $i$th column is given by $\hat{f}(i)$ for all $i \in [n]$ and let $Y = x^Tx$. We can now calculate $z_P(Y)$:
$$z_P(Y) 
     = \sum_{(i,j) \in E} \frac{1-Y_{ij}}{2} \\
    = \sum_{(i,j) \in E}  \frac{1-\hat{f}(i) \cdot \hat{f}(j)}{2}$$
    $$
    = \sum_{(i,j) \in E} \frac{1-f(i) \cdot f(j)}{2} = \sum_{(i,j) \in E} \frac{1-(\xi_2/k)}{2} \\
    = \frac{|E|}{2}(1-\xi_2/k). \qedhere$$

\end{proof}

\begin{prop}
\label{thm:feasibleDualSRG}
Let $G$ be a $\text{SRG}(n,k,\lambda,\mu)$ where $n=4(3t+1), k=3(t+1), \lambda=2, \mu=t+1$ for some non-negative integer $t$. Then there exists a feasible $\zeta$ for the dual of the GW SDP so that $z_D(\zeta) = \frac{|E|}{2}(1-\xi_2/k)$.
\end{prop}
\begin{proof}
    Let $G$ be a $\text{SRG}(n,k,\lambda,\mu)$. Let $\xi_2$ be the smallest eigenvalue of the adjacency matrix $A$ of $G$. The dual of the GW SDP relaxation for unit-weight graphs is given as follows: find $\zeta = (\zeta_1,\dots, \zeta_n) \in \mathbb{R}^n$ to minimize
$$z_D(\zeta) = \frac{|E|}{2}+\frac{1}{4}\sum_{i=1}^n \zeta_i,$$
subject to
$A + \text{diag}(\zeta_1,\dots,\zeta_n)$
being positive semidefinite, where $\text{diag}(\zeta_1,\dots,\zeta_n)$ is a diagonal matrix whose $i$th entry on the diagonal is $\zeta_i$ \cite{K99}.
    
    Let $\zeta = (-\xi_2, \dots, -\xi_2)$. Then,
    $A+\text{diag}(\zeta_1,\dots,\zeta_n)= A-\xi_2 I.$
    It is a standard result of linear algebra that for any square matrix $M$ and constant $c$, if $\lambda$ is an eigenvalue of $M$, then $\lambda+c$ is an eigenvalue of $M+cI$; thus letting $\lambda_\text{min}(M)$ denote the smallest eigenvalue of a matrix $M$, we have,
    $$\lambda_\text{min}(A+\text{diag}(\zeta_1,\dots,\zeta_n)) = \lambda_\text{min}(A - \xi_2 I) = \xi_2 - \xi_2 = 0.$$

    Since $\lambda_\text{min}(A+\text{diag}(\zeta_1,\dots,\zeta_n)) \geq 0$, then $A+\text{diag}(\zeta_1,\dots,\zeta_n)$ is positive-semidefinite, meaning that $\zeta$ is a feasible solution.

    We now calculate $z_D(\zeta)$:
    \begin{align*}
        z_D(\zeta) 
        &= \frac{|E|}{2} + \frac{1}{4}\sum_{i=1}^n \zeta_i \\
        &= \frac{|E|}{2} + \frac{1}{4}\sum_{i=1}^n -\xi_2 = \frac{|E|}{2} - \frac{\xi_2 n}{4} \\
        &= \frac{|E|}{2}\left(1 - \frac{2}{|E|}\frac{\xi_2 n}{4}\right)\\
        &= \frac{|E|}{2}\left(1 - \frac{2}{nk/2}\frac{\xi_2 n}{4}\right) \tag{$|E| = nk/2$}\\
        &= \frac{|E|}{2}\left(1 -\frac{\xi_2}{k}\right), \text{     as desired.} \qedhere
        % &= \frac{|E|}{2} + k\frac{n}{12} \tag{$\xi_2 = -k/3$ from proof of Prop. \ref{thm:angleForQ3t}} \\
        % &= \frac{|E|}{2} + \frac{2|E|}{n}\frac{n}{12} \tag{$|E| = nk/2$ as $G$ is $k$-regular} \\
        % &= \frac{2}{3}|E|, 
    \end{align*} 
\end{proof}

\begin{theorem}
\label{thm:optimality}
     Let $G$ be a $\text{SRG}(n,k,\lambda,\mu)$ and let $\xi_2$ be the smallest eigenvalue of the adjacency matrix of $G$. Then $z_P^* = \frac{|E|}{2}(1-\xi_2/k)$ (and similarly, $z_D^* = \frac{|E|}{2}(1-\xi_2/k)$). Moreover, the mapping $f$ obtained by Corollary \ref{thm:srgToSphere} corresponds to an optimal GW SDP solution $Y$.
\end{theorem}
\begin{proof}
    Let $G$ be a $\text{SRG}(n,k,\lambda,\mu)$ and  let $Y$ be the feasible GW SDP solution corresponding to the mapping $f$ obtained in Corollary \ref{thm:angleForQ3t}. From Propositions \ref{thm:feasiblePrimalSRG} and \ref{thm:feasibleDualSRG}, there exists feasible $\zeta$ such that $z_P(Y) = \frac{|E|}{2}(1-\xi_2/k) = z_D(\zeta)$. Thus,
    $$\frac{|E|}{2}(1-\xi_2/k) = z_P(Y) \leq z_P^* \overset{(a)}{\leq} z_D^* = z_D(\zeta) =\frac{|E|}{2}(1-\xi_2/k),$$
    % implying that $\frac{|E|}{2}(1-\xi_2/k) \leq z_P^* \leq \frac{|E|}{2}(1-\xi_2/k)$ and $\frac{|E|}{2}(1-\xi_2/k) \leq z_D^* \leq \frac{|E|}{2}(1-\xi_2/k)$, 
    where (a) follows from strong duality\footnote{Here, strong duality (i.e. $z_P^* = z_D^*$) holds due to  Slater's conditions \cite{boyd2004convex}.}. Therefore, $z_P^* =\frac{|E|}{2}(1-\xi_2/k)= z_P(Y)$ and $z_D^* = \frac{|E|}{2}(1-\xi_2/k) = z_D(\zeta)$. Note that in the inequalities above, $z_P^* \leq z_D^*$ follows by weak duality. \qedhere
\end{proof}

From Theorem \ref{thm:optimality}, we see that the optimal SDP value is dependent on $\xi_2$, the smallest eigenvalue of the adjacency matrix of the graph. It is well-known that the eigenvalues of a strongly-regular graph are entirely determined by its parameters ($n,k,\lambda,\mu$); an explicit listing of these eigenvalues can be seen in Theorem \ref{thm:srgEigenvalues} below.

\begin{theorem}[Seidel \ \cite{seidel1979strongly}]
\label{thm:srgEigenvalues}
    Let $G$ be an $\text{SRG}(n,k,\lambda,\mu)$. Then there are up to three distinct eigenvalues: $k, \xi_1, \xi_2$, with $\xi_1 \geq \xi_2 \geq 1$ and
    
    $$\{\xi_1, \xi_2\} = \frac{1}{2}\left[(\lambda-\mu) \pm \sqrt{(\lambda-\mu)^2 + 4(k-\mu)} \right];$$
    the multiplicities of these eigenvalues are a function of $n,k,\lambda, \mu$.
    
    % Then the eigenvalues of the adjacency matrix of $G$ are:
    % \begin{itemize}
    %     \item Eigenvalue $k$ with multiplicity 1.
    %     \item Eigenvalue $$\xi_1 = \frac{1}{2}\left[(\lambda-\mu) + \sqrt{(\lambda-\mu)^2 + 4(k-\mu)} \right]$$ with multiplicity $$\frac{1}{2}\left[(n-1) - \frac{2k+(n-1)(\lambda-\mu)}{\sqrt{(\lambda-\mu)^2 + 4(k-\mu)}} \right].$$
    %     \item Eigenvalue $$\xi_2 = \frac{1}{2}\left[(\lambda-\mu) - \sqrt{(\lambda-\mu)^2 + 4(k-\mu)} \right]$$ with multiplicity $$\frac{1}{2}\left[(n-1) + \frac{2k+(n-1)(\lambda-\mu)}{\sqrt{(\lambda-\mu)^2 + 4(k-\mu)}} \right].$$
    % \end{itemize}

    % Additionally, $\xi_2$ is the smallest eigenvalue.
\end{theorem}

In the next section, we use the above results to obtain concrete approximation ratios on strongly-regular graphs with a specific choice of parameters.

\subsection{SRGs with Specific Parameters}
We now consider the family of strongly-regular graphs whose parameters are given by $n=4(3t+1), k=3(t+1), \lambda=2, \mu=t+1$ for some integer $t$.

This choice of parameters for this family of strongly-regular graphs was inspired by particular instances in the MQLib library \cite{DGS18} and ultimately found due to the connection that strongly-regular graphs have with partial geometries; more precisely, it can be shown that the family of strongly-regular graphs we consider are so-called pseudo-geometric graphs for generalized quadrangles of order $(3,t)$ for some positive integer $t$ \cite{haemers2001pseudo}. This family of graphs is finite, with a total of 167 non-isomorphic graphs \cite{haemers2001pseudo}; see Table \ref{tab:q3tInstances} for a table of all possible instances in this family along with their parameters. Additionally, this family of strongly-regular graphs has several instances under 100 nodes, making them suitable for benchmarking purposes on current and near-term quantum devices.

\begin{table}
    \centering
    
    \begin{tabular}{|c||c|c|c|c||c|c|}
        \hline $t$ & $n$ & $k$ & $\lambda$ & $\nu$ & \# of instances  & $\alpha_\text{QAOA}$ \\ \hline 
        1 & 16 & 6 & 2 & 2 & 2 & 0.8935\\ 
        3 & 40 & 12 & 2 & 4 & 28 & 0.8605\\ 
        5 & 64 & 18 & 2 & 6 & 167 & 0.8433\\ 
        9 & 112 & 30 & 2 & 10 & 1 & 0.8246 \\\hline
    \end{tabular}
    \caption{\label{tab:q3tInstances} \footnotesize All possible strongly-regular graphs of the form $n=4(3t+1), k=3(t+1), \lambda=2, \mu=t+1$ for some integer $t$ along with the number of non-isomorphic graphs found within each parameter set (here, $k$ is the degree of the vertices). All graphs satisfy $\mc(G) = \frac{2}{3}|E|$ with the exception of 13 instances (see Table \ref{tab:stronglyRegularExceptions} for the approximations achieved on these instances). Excluding these exceptions, the last column denotes the approximation ratio achieved by depth-1 QAOA (see Section \ref{sec:qaoaPerformanceInterestingInstances}).}
\end{table}

These parameter choices induces a family of graphs whose 2nd smallest eigenvalue is given by the simple expression $(\xi_2 = -k/3$); by Corollary \ref{thm:angleForQ3t}, this implies that, for this family of graphs, there exists a spherical mapping $f$ of the vertices such that the angles between adjacent vertices are equal (with angle $\arccos(-1/3) \approx 109.47^\circ$). We remark that this angle is identical to the angles formed by the vertices of a regular tetrahedron with its center.\footnote{See \cite{brittin1945valence} for a simple, yet succinct proof that such angles in a tetrahedron have a measure of $109.47^\circ$.}

\begin{prop}
\label{thm:eigenForQ3t}
    Let $G$ be a $\text{SRG}(n,k,\lambda,\mu)$ where $n=4(3t+1), k=3(t+1), \lambda=2, \mu=t+1$ for some non-negative integer $t$. Let $\xi_2$ be the smallest eigenvalue of the adjacency matrix of $G$. Then $\xi_2 = -\frac{k}{3}$.
\end{prop}
\begin{proof}
    Observe that for such a $G$, the SRG parameter $\mu$ is equal to $\mu = t+1 = \frac{1}{3}\cdot 3(t+1) = \frac{k}{3}$.

    From Proposition \ref{thm:srgEigenvalues}, substitution of the above parameter values, and some algebraic manipulation, we have that
    $$\xi_2 = \frac{1}{2}\left[(\lambda-\mu) - \sqrt{(\lambda-\mu)^2 + 4(k-\mu)} \right] $$$$=  \frac{1}{2}\left[(2-k/3) - \sqrt{(2-k/3)^2 + 4(k-k/3)} \right] = -\frac{k}{3}. \qedhere$$ 

\end{proof}

\begin{corollary}
\label{thm:angleForQ3t}
    Let $G=(V,E)$ be a $\text{SRG}(n,k,\lambda,\mu)$ where $n=4(3t+1), k=3(t+1), \lambda=2, \mu=t+1$ for some non-negative integer $t$. Then there exists an integer $r\leq n$ and a function $f: V \to \mathbb{S}^{r-1}$ so that for all $(u,v) \in E$, $f(u) \cdot f(v) = -1/3$. Moreover, the mapping $f$ obtained by Corollary \ref{thm:srgToSphere} corresponds to an optimal GW SDP solution $Y$.
\end{corollary}
\begin{proof}
     By Proposition \ref{thm:angleForQ3t}, the smallest eigenvalue of the adjacency matrix of such a $G$ is given by $\xi_2 = -k/3$. Note that any non-primitive strongly-regular graph must either satisfy $\mu = 0$ or $\mu = k$; since this is not the case for $G$, then $G$ must be a primitive strongly-regular graph. Thus, by Theorem \ref{thm:srgToSphere}, there exists an integer $r\leq n$ and a function $f: V \to \mathbb{S}^{r-1}$ so that for all $(u,v) \in E$, $f(u) \cdot f(v) = \frac{\xi_2}{k} = \frac{-k/3}{k} = -\frac{1}{3}$, and by Theorem \ref{thm:optimality}, this mapping $f$ corresponds to an optimal SDP solution.
\end{proof}

\begin{corollary}
\label{thm:q3tOptimality}
     Let $G=(V,E)$ be a $\text{SRG}(n,k,\lambda,\mu)$ where $n=4(3t+1), k=3(t+1), \lambda=2, \mu=t+1$ for some non-negative integer $t$. Then $z_P^* = \frac{2}{3}|E|$ (and similarly, $z_D^* = \frac{2}{3}|E|$). Moreover, the mapping $f$ obtained by Corollary \ref{thm:angleForQ3t} corresponds to an optimal GW SDP solution $Y$.
\end{corollary}
\begin{proof}
    Recall that for such SRG, that $\xi_2 = -k/3$ (Proposition \ref{thm:eigenForQ3t}). Then the result follows from Theorem \ref{thm:q3tOptimality} as $z_P^* = z_D^* = \frac{|E|}{2}(1-\xi_2/k) = \frac{|E|}{2}(1 - (-k/3)/k) = \frac{2}{3}|E|$.
\end{proof}

Next, we discuss the Max-Cut value that is achieved for this family of strongly regular graphs. Since the Max-Cut value is upper-bounded by the optimal SDP value, then by Corollary \ref{thm:q3tOptimality}, we have that $\mc(G) \leq \frac{2}{3}|E|$ for all graphs $G$ in this family.
%It is known that that there are a finite number of strongly-regular graphs with such parameters \cite{haemers2001pseudo}. 
A straightforward computer verification shows that for all but 13 instances in this (finite) family of strongly-regular graphs, the maximum cut contains exactly $\frac{2}{3}|E|$ edges. A list of exceptions can be found in Table \ref{tab:stronglyRegularExceptions}; we remark that all the exceptions have 40 nodes (corresponding to $t=3$) and that the maximum cut value is still close to (but not equal to) two-thirds of the number of edges. The verification code can be found online in the CI-QuBe library \cite{CI-QuBe2021}. We show that this family of graphs has approximation ratio of at least $0.912$ (with respect to the GW) algorithm, with the approximation ratio being exact in the case where $\mc(G) = \frac{2}{3}|E|$.

\begin{table}
    \centering

    \begin{tabular}{|c|c|c|c|c|c|c|}
    \hline
        ID \#'s & $n$ & $|E|$ & $\text{MC}(G)$ & $\frac{\text{MC}(G)}{|E|}$ & $\alpha_{G, \text{GW}}$ & $\alpha_\text{QAOA}$  \\ \hline
11 -- 15 & 40 & 240 & 156 & 0.6500 & 0.9357 & $0.8826$  \\ \hline
16 -- 23 & 40 & 240 & 158 & 0.6583 & 0.9238 & $0.8714$  \\ \hline
    \end{tabular}
\caption{\label{tab:stronglyRegularExceptions} \footnotesize The 13 instances $G = (V,E)$ of strongly-regular graphs parameterized by $n=4(3t+1), k = 3(t+1), \lambda = 2, \mu = t+1$ for some $t$, where $\frac{\text{Max-Cut}(G)}{|E|} \neq \frac{2}{3}$. The last two columns are the instance-specific approximation ratio that the GW and QAOA algorithms achieve on each of these instances.}
    
\end{table}

\iffalse \sout{Finally, we prove that for the family of strongly-regular graphs considered in this subsection (i.e. those parametrized by $(4(3t+1),3(t+1),2,t+1)$ for some $t$), the maximum cut contains two-thirds of the edges in the graph. Put another way, since the optimal objective value of the GW SDP is $\frac{2}{3}|E|$ for this family of graphs, our proof is equivalent to proving that the GW SDP relaxation has an integrality gap of 1 when restricted to this family of strongly-regular graphs.}

\begin{prop}
\label{thm:srgMaxcut}
\sout{Let $G=(V,E)$ be a $\text{SRG}(n,k,\lambda,\mu)$ where $n=4(3t+1), k=3(t+1), \lambda=2, \mu=t+1$ for some non-negative integer $t$. Then $\text{\mc(G)}=\frac{2}{3}|E|$.}
\end{prop}
\begin{proof}
    \sout{It can be shown that there are a finite number of strongly-regular graphs whose SRG parameters have the form $n=4(3t+1), k=3(t+1), \lambda=2, \mu=t+1$ \cite{haemers2001pseudo}; a simple computer verification demonstrates that the proposition holds for each of these instances.}
\end{proof}

\sout{A non-computer proof of Proposition \ref{thm:srgMaxcut} above would be more satisfying; however, we were unable to find such a proof as of the writing of this work. We believe that such a non-computer proof and the other proofs in this sub-section have the potential to be generalized to produce theorems similar to Theorem \ref{thm:q3tAR} below for other families of strongly-regular graphs.}
\fi 

\begin{theorem}
\label{thm:q3tAR}
    Let $G$ be a $\text{SRG}(n,k,\lambda,\mu)$ where $n=4(3t+1), k=3(t+1), \lambda=2, \mu=t+1$ for some non-negative integer $t$. The GW algorithm achieves an instance-specific approximation ratio of at least $0.912$ on these instances. Moreover, if $\text{Max-Cut}(G) = \frac{2}{3}|E|$, then the instance-specific approximation ratio is 0.912.
\end{theorem}
\begin{proof}
    Let $G=(V,E)$ be such a graph. By Corollary \ref{thm:q3tOptimality}, we have that $\mc(G) \leq z_p^* = \frac{2}{3}|E|$. By Corollary \ref{thm:angleForQ3t}, there exists an $r \leq n$ and $f:V \to \mathbb{S}^{r-1}$ so that $\arccos(f(u) \cdot f(v)) = \arccos(-1/3)$ for all $(u,v)\in E$; let $\theta = \arccos(-1/3)$. By Theorem \ref{thm:optimality}, this mapping corresponds to an optimal solution $Y$ to the GW SDP relaxation. Since $Y$ is optimal with respect to the GW SDP, then by Proposition \ref{thm:equalAngles}, the instance-specific GW approximation ratio for $G$ is given by:
    $$\frac{\frac{\theta}{\pi}|E|}{\text{\mc}(G)} \geq \frac{\frac{\theta}{\pi}|E|}{\frac{2}{3}|E|} = \frac{3}{2}\cdot \frac{\theta}{\pi} =  \frac{3}{2}\cdot \frac{\arccos(-1/3)}{\pi}\approx 0.912;$$
    when $\mc(G) = \frac{2}{3}|E|$, then %\sout{inequality above becomes an inequality, giving that the} 
    we get that the instance-specific approximation ratio is $\approx 0.912$.% in such a case.
\end{proof}

For the instances for which $\text{Max-Cut} = \frac{2}{3}|E|$ holds, a non-computer proof of this equality would be more satisfying; however, we were unable to find such a proof as of the writing of this work. We believe that such a non-computer proof and the other proofs in this sub-section have the potential to be generalized to produce theorems similar to Theorem \ref{thm:q3tAR} below for other families of strongly-regular graphs.

\section{QAOA's Performance on Challenging Instances for the GW Algorithm}
\label{sec:qaoaPerformanceInterestingInstances}

The GW algorithm does not provide optimal cuts for the instances found using Karloff's construction; one may hope that quantum computing may have some advantage over these instances which are, in a sense, classically difficult. We show in Theorem \ref{thm:karloffQAOAApproxRatio}, that depth-1 QAOA achieves a lower instance-specific approximation ratio for the instances that arise from Karloff's construction; this implies that, for these instances, either a higher circuit depth or some modification of the QAOA algorithm would be needed in order to demonstrate some form of quantum advantage (over the GW algorithm).

To first make the sequence of graphs considered more precise, consider the family of Karloff instances $G_m = J(m,m/2,b)$ where $m$ is even and $b = \lceil \frac{1}{4}(\cos(\theta^*)+1)m\rceil$  and $\theta^* = \text{argmin}_{0\leq \theta \leq \pi} \frac{2}{\pi}\frac{\theta}{1-\cos \theta} \approx 2.33112$; recall from Section \ref{sec:GW_Alg} that this choice of $b$ is (near) optimal in the sense that, for fixed $m$, the corresponding graph will have the worst possible approximation ratio with respect to the GW algorithm. Some calculations give us that $b \approx \lceil 0.077m\rceil$. We will assume that $m\geq 12$, and hence one can show that $0 \leq b < \frac{m}{6}$. It follows from Karloff \cite{K99} that the expected cut value of $G_m$ approaches $\alpha^* = 0.878$ as $m\to \infty$.

We first prove two helpful lemmas. In the first lemma, we prove that the degree of $G_m$ is ${m/2 \choose b}^2$ and in the next lemma, we prove that for $m$ large enough, the graphs $G_m$ are triangle free.
\begin{lemma}\label{thm:karloffDegree}
The degree of each vertex in $G_m$ is  ${m/2 \choose b}^2$.
\end{lemma}
\begin{proof}
Let us now determine the degree of $G_m$. Let $S$ be a vertex of $G_m$. To construct a neighbor $T$, we must pick $b$ elements from the $m/2$ elements in $S$ for $T$ for the overlap, then from the remaining $m/2$ elements in $[m] \setminus S$, we must pick $m/2 - b$ elements to determine the remaining elements in $T$. Thus, the neighbor of ways to construct a neighbor for $S$ are,
$\text{deg}(S) = {m/2 \choose b}{m/2 \choose m/2 - b} = {m/2 \choose b}^2.$
\end{proof}

\begin{lemma}\label{thm:karloffTriangleFree}
    For all $m\geq 12$, $G_m$ is triangle-free.
\end{lemma}
\begin{proof}
As $m\geq 12$, we have that $0 \leq b < m/6$ as determined previously. Now, for a given edge $(S,T)$ in $G_m$, let us count how many ways we can construct a mutual neighbor $U$ of both $S$ and $T$. Let $k$ be the number of elements in $U$ that are also contained in $S \cap T$. Then, by combinatorially counting, the number of ways to construct a mutual neighbor $U$ of both $S$ and $T$ are: 
$$\sum_{k=0}^b {b \choose k}{m/2 - b \choose b-k}^2{b \choose m/2-2b+k}.$$ 

We claim that the above sum is zero. To see this, consider the term ${b \choose m/2-2b+k}$ in the sum above. Observe that, for $k\geq 0$ and since $0\leq b < m/6$, we have
$m/2 - 2b+k \geq m/2-2b > (6b)/2-2b = b$
and hence the ${b \choose m/2-2b+k}$ term is zero. \end{proof}%, making the whole sum zero.

%Since $|S \cap T| = b$, there are ${b \choose k}$ ways to pick such $k$ elements for $U$. Now, $|S \cap U| = b$, we need to pick $b-k$ more elements from the remaining $m/2-b$ elements of $S$ to add to $U$; there are ${m/2-b \choose b-k}$ ways to do this. Similarly, there are ${m/2-b \choose b-k}$ ways to choose the elements of $T$ (not in $S \cap T$) that also belong in $U$. Lastly, there are $m/2 - (k + (b-k) + (b-k)) = m/2 - 2b+k$ remaining elements we need to add to $U$ to ensure it has a total of $m/2$ elements and these can be selected from the $|[m] \setminus (S \cup T)| = m-(m/2+m/2-b) = b$ elements outside of $S \cup T$; there are ${b \choose m-2b+k}$ ways to do this. Summing over the possible choices of $k=0,\dots,b$, we have that the total number of mutual neighbors between $S$ and $T$ are
%$$\sum_{k=0}^b {b \choose k}{m/2 - b \choose b-k}^2{b \choose m/2-2b+k}.$$

%We claim that the above sum is zero. To see this, consider the term ${b \choose m/2-2b+k}$ in the sum above. Observe that, for $k\geq 0$ and since $0\leq b < m/6$, we have $m/2 - 2b+k \geq m/2-2b > (6b)/2-2b = b$ and hence the ${b \choose m/2-2b+k}$ term is zero, making the whole sum zero.

Finally, we show that as $m$ increases, the approximation ratio that depth-1 QAOA achieves on $G_m$ approaches 0.592.

\begin{theorem}
\label{thm:karloffQAOAApproxRatio}
Let $F_p(G)$ denote the expected cut value obtained from depth-$p$ (standard) QAOA when the optimal choice of variational parameters $(\gamma,\beta)$ are used  over the family of Karloff instances $\{G_m\}$. Then $$\lim_{m \to \infty} \frac{F_1(G_m)}{\text{Max-Cut}(G)} = 0.592.$$
\end{theorem}
\begin{proof}
By \cite{WHJR18}, since $G_m$ is triangle-free (for $m\geq 12$) due to Lemma \ref{thm:karloffTriangleFree}, we have that depth-1 QAOA achieves the following expected cut value when the optimal parameters
% , i.e., $$(\gamma,\beta) = (\text{arctan}(1/\sqrt{d-1}),\pi/8),$$ 
are used:
$$F^*(G_m) = \frac{|E|}{2}\left(1+\frac{1}{\sqrt{d}}\left(\frac{d-1}{d}\right)^{(d-1)/2}\right);$$

here, $d$ is the degree of the (regular) graph $G_m$, which is $d={m/2 \choose b}^2$ (Lemma \ref{thm:karloffDegree}). From Karloff's work \cite{K99}, since $0 \leq b < m/6 < m/4$, we have that the maximum cut is given by
\begin{equation}\label{eqn:maxcut_GW}\text{Max-Cut}(G_m) = \frac{n}{2}{m/2 \choose b}^2\left(1-\frac{2b}{m}\right).\end{equation}
Thus, the approximation ratio for depth-1 standard QAOA on $G_m$ (with $m\geq 12$) goes to, \begin{align*}
   &\lim_{m \to \infty} F^*(G_m)/\text{Max-Cut}(G_m) = \lim_{m \to \infty} \frac{\frac{|E|}{2}\left(1+\frac{1}{\sqrt{d}}\left(\frac{d-1}{d}\right)^{(d-1)/2}\right)}{\frac{n}{2}{m/2 \choose b}^2\left(1-\frac{2b}{m}\right)} \\
    &= \lim_{m \to \infty} \frac{\frac{nd/2}{2}\left(1+\frac{1}{\sqrt{d}}\left(\frac{d-1}{d}\right)^{(d-1)/2}\right)}{\frac{n}{2}{m/2 \choose b}^2\left(1-\frac{2b}{m}\right)} \tag{as $G_m$ is $d$-regular}\\
    &= \lim_{m \to \infty} \frac{\frac{nd/2}{2}}{\frac{n}{2}{m/2 \choose b}^2\left(1-\frac{2b}{m}\right)} \tag{ $\left(\frac{d-1}{d}\right)^{(d-1)/2} \to e^{-1/2}$ as $d \to \infty$}\\
    &= \lim_{m \to \infty} \frac{\frac{nd/2}{2}}{\frac{n}{2}d\left(1-\frac{2b}{m}\right)} \tag{as $d = {m/2 \choose b}^2$}\\
    % \end{align*}
    % \begin{align*}
%    &= \lim_{m \to \infty} \frac{1}{1+1-\frac{4b}{m}} \tag{simplify}\\
    &=  \frac{1}{1-\cos(\theta^*)} \tag{as $\cos(\theta^*) = \frac{4b}{m}-1$ as $m \to \infty$}\\
    &= \frac{\pi}{2}\frac{\alpha^*}{\theta^*} \approx 0.592 \tag{as $\alpha^* = \frac{2}{\pi}\frac{\theta^*}{1-\cos \theta^*}$}.
\end{align*}
In the above calculations we use that as $m\to \infty$, then $b \to \infty$ and we have that (for $m$ large enough)\footnote{One can show that $b$ is strictly positive when $m \geq 12$. If we instead take $b = \lfloor 0.077 m \rfloor$, then $b$ is strictly positive whenever $m \geq 14$.}, $d = {m/2 \choose b} \geq {m/2 \choose 1} = m/2 \to \infty$. The calculations also hold true (including the triangle-free result) if we consider $b = \lfloor 0.077 m \rfloor$ instead of $b = \lceil 0.077 m \rceil$ in the construction of $G_m$.
\end{proof} 

The above theorem considers Karloff instances where GW obtains a close to 0.878 approximation (i.e., $b/m \approx 0.077$). In \ref{sec:other_karloff_instances}, we discuss the performance gap of GW and QAOA on other instances with varying $b/m$; essentially the gap goes to zero as $b/m \rightarrow 1/4$, where the instances become easy for both algorithms. 

In the case of the family of strongly-regular graphs considered in Section \ref{sec:stronglyRegularApprox} (parameterized by $n=4(3t+1), k=3(t+1), \lambda=2, \mu=t+1$ for some non-negative integer $t$), the proof used in Theorem \ref{thm:karloffQAOAApproxRatio} can not be used as the proof depends on the graphs being triangle-free and the family of strongly-regular graphs that we consider are not triangle-free (as $\lambda = 2 \neq 0$). Instead, we compute the approximation ratio of QAOA attained at depth $p=1$, using the following theorem:

\begin{theorem}[\cite{WHJR18}]
\label{thm:probOfCuttingEdgeFormula}
For QAOA with $p=1$ layers and variational parameters $\gamma$ and $\beta$, the probability of cutting an edge $(u,v)$ is
$\frac{1}{2} + \frac{1}{4}(\sin 4\beta \sin \gamma)(\cos^{d_u} \gamma + \cos^{d_v} \gamma) - \frac{1}{4}(\sin^2 2\beta \cos^{d_u +d_v - 2\lambda_{uv}} \gamma)(1 - \cos^{\lambda_{uv}}2\gamma),$
where $d_u+1$ and $d_v+1$ are the degrees of vertices $u$ and $v$ respectively, and $\lambda_{uv}$ is the number of common neighbors of $u$ and $v$.
\end{theorem}

For strongly-regular graphs with parameters $(n,k,\lambda, \mu)$, the expected cut value can then be calculated by setting $d_u = d_v = k - 1$ and $\lambda_{uv} = \lambda$ for all $u,v \in E$ in the formula above and multiplying by the total number of edges $|E| = nk/2$. For each possible $t \in \{1,3,5,9\}$ that parametrizes the strongly-regular instances in this work, we optimize the $\gamma$ and $\beta$ quantum circuit parameters of Theorem \ref{thm:probOfCuttingEdgeFormula} by doing a fine $5000 \times 5000$ grid search over the domain $\gamma \in [-\pi/2, \pi/2)$ and $\beta \in [-\pi/4, \pi/4)$; globally optimal circuit parameters are guaranteed to exist in this domain for any unit-weight regular graph \cite{ZWCPL20}.

We note that $p=1$ QAOA achieves an approximation ratio in $[0.82, 0.90]$ over all the strongly-regular graphs considered in this work, see Tables \ref{tab:q3tInstances} and \ref{tab:stronglyRegularExceptions}. These theorems indicate that a higher circuit depth and/or a modification to QAOA is needed in order to surpass the performance of the GW algorithm.

\section{Discussion}
\label{sec:discussion}

Identifying challenging problem instances across a range of sizes is critical for the effective benchmarking and comparison of optimization algorithms. This need is particularly pronounced in the NISQ (Noisy Intermediate-Scale Quantum) era of quantum computing, where hardware limitations constrain the size of solvable instances. In such settings, many small-scale instances tend to be trivially solvable by classical methods, often achieving near-optimal approximation ratios in minimal time. This hampers the ability to meaningfully differentiate algorithmic performance.

In this work, we focused on the performance of depth-1 QAOA and GW for solving Max-Cut over certain strongly regular graph families. In preliminary computational exploration, we found that Karloff instances with modified edge weights also result in an approximation ratio $<0.95$ for the Goemans-Williamson algorithm. Additionally, we evaluated a diverse set of instances from the MQLib library \cite{DGS18} using their suite of 37 implemented heuristics. We observed that the most challenging graphs have both positive and negative edge weights spanning several orders of magnitude (see Figure \ref{fig:edgeWeights} in the appendix for an example graph from \href{https://github.com/MQLib/MQLib}{MQLib}). We include these instances into \href{link}{CI-QuBe} github repository to enable future work, see  \ref{sec:mixedWeightInstances} for details. We hope that this work will help both classical and quantum optimization communities in furthering our understanding on the strengths and weaknesses of classical and quantum algorithms.

\section{Acknowledgements}
The research presented in this article was supported by the Laboratory Directed Research and Development program of Los Alamos National Laboratory under project number 20230049DR as well as by the NNSA's Advanced Simulation and Computing Beyond Moore's Law Program at Los Alamos National Laboratory. This work has been assigned LANL technical report number LA-UR-25-28776. The authors also acknowledge support from  DARPA
under Contract No. HR001120C0046, when both authors were at Georgia Institute of Technology.

% \rt{Mention Ci-Qube library. Perturbed Karloff instances. Mixed Weight Graphs. We believe this gap in the literature might be worth exploring, for finding instances where both classical and quantum algorithms have difficulty. We hope that this work is of interest to both communities of classical and quantum optimization.}



\bibliographystyle{elsarticle-num} 
\bibliography{references}

\def\authornoop#1{}
\begin{thebibliography}{10}
\expandafter\ifx\csname url\endcsname\relax
  \def\url#1{\texttt{#1}}\fi
\expandafter\ifx\csname urlprefix\endcsname\relax\def\urlprefix{URL }\fi
\expandafter\ifx\csname href\endcsname\relax
  \def\href#1#2{#2} \def\path#1{#1}\fi

\bibitem{Google2019}
F.~Arute, K.~Arya, R.~Babbush, D.~Bacon, J.~C. Bardin, R.~Barends, R.~Biswas, S.~Boixo, F.~G. S.~L. Brandao, D.~A. Buell, B.~Burkett, Y.~Chen, Z.~Chen, B.~Chiaro, R.~Collins, W.~Courtney, A.~Dunsworth, E.~Farhi, B.~Foxen, A.~Fowler, C.~Gidney, M.~Giustina, R.~Graff, K.~Guerin, S.~Habegger, M.~P. Harrigan, M.~J. Hartmann, A.~Ho, M.~Hoffmann, T.~Huang, T.~S. Humble, S.~V. Isakov, E.~Jeffrey, Z.~Jiang, D.~Kafri, K.~Kechedzhi, J.~Kelly, P.~V. Klimov, S.~Knysh, A.~Korotkov, F.~Kostritsa, D.~Landhuis, M.~Lindmark, E.~Lucero, D.~Lyakh, S.~Mandrà, J.~R. McClean, M.~McEwen, A.~Megrant, X.~Mi, K.~Michielsen, M.~Mohseni, J.~Mutus, O.~Naaman, M.~Neeley, C.~Neill, M.~Y. Niu, E.~Ostby, A.~Petukhov, J.~C. Platt, C.~Quintana, E.~G. Rieffel, P.~Roushan, N.~C. Rubin, D.~Sank, K.~J. Satzinger, V.~Smelyanskiy, K.~J. Sung, M.~D. Trevithick, A.~Vainsencher, B.~Villalonga, T.~White, Z.~J. Yao, P.~Yeh, A.~Zalcman, H.~Neven, J.~M. Martinis, Quantum supremacy using a programmable superconducting processor, Nature 574 (2019) 505--510.

\bibitem{S94}
P.~W. Shor, Algorithms for quantum computation: discrete logarithms and factoring, in: Proceedings 35th Annual Symposium on Foundations of Computer Science, 1994, pp. 124--134.

\bibitem{HTOLHS21}
R.~Herrman, L.~Treffert, J.~Ostrowski, P.~C. Lotshaw, T.~S. Humble, G.~Siopsis, Impact of graph structures for qaoa on maxcut, Quantum Information Processing 20~(9) (2021) 1--21.

\bibitem{ZWCPL20}
L.~Zhou, S.-T. Wang, S.~Choi, H.~Pichler, M.~D. Lukin, Quantum approximate optimization algorithm: Performance, mechanism, and implementation on near-term devices, Physical Review X 10~(2) (2020) 021067.

\bibitem{bayerstadler2021industry}
A.~Bayerstadler, G.~Becquin, J.~Binder, T.~Botter, H.~Ehm, T.~Ehmer, M.~Erdmann, N.~Gaus, P.~Harbach, M.~Hess, et~al., Industry quantum computing applications, EPJ Quantum Technology 8~(1) (2021) 25.

\bibitem{campbell2022qaoa}
C.~Campbell, E.~Dahl, Qaoa of the highest order, in: 2022 IEEE 19th International Conference on Software Architecture Companion (ICSA-C), IEEE, 2022, pp. 141--146.

\bibitem{FH16}
E.~Farhi, A.~W. Harrow, Quantum supremacy through the quantum approximate optimization algorithm, arXiv preprint arXiv:1602.07674 (2016).

\bibitem{shaydulin2021classical}
R.~Shaydulin, S.~Hadfield, T.~Hogg, I.~Safro, Classical symmetries and the quantum approximate optimization algorithm, Quantum Information Processing 20 (2021) 1--28.

\bibitem{FGG14}
E.~Farhi, J.~Goldstone, S.~Gutmann, A quantum approximate optimization algorithm, arXiv preprint arXiv:1411.4028 (2014).

\bibitem{dunning2018}
I.~Dunning, S.~Gupta, J.~Silberholz, What works best when? {A} systematic evaluation of heuristics for {Max-Cut} and {QUBO}, INFORMS Journal on Computing 30~(3) (2018) 608--624.

\bibitem{GW95}
M.~X. Goemans, D.~P. Williamson, Improved approximation algorithms for maximum cut and satisfiability problems using semidefinite programming, Journal of the ACM (JACM) 42~(6) (1995) 1115--1145.

\bibitem{Trevisan12}
L.~Trevisan, Max cut and the smallest eigenvalue, SIAM Journal of Computing 41 (2012) 1769--1786.

\bibitem{Soto15}
J.~A. Soto, Improved analysis of a max-cut algorithm based on spectral partitioning, SIAM Journal on Discrete Mathematics 29 (2015) 259--–268.

\bibitem{lasserre2001global}
J.~B. Lasserre, Global optimization with polynomials and the problem of moments, SIAM Journal on optimization 11~(3) (2001) 796--817.

\bibitem{parrilo2000structured}
P.~A. Parrilo, Structured semidefinite programs and semialgebraic geometry methods in robustness and optimization, California Institute of Technology, 2000.

\bibitem{Khot02}
S.~Khot, On the power of unique 2-prover 1-round games, in: Proceedings of the thiry-fourth annual ACM symposium on Theory of computing, 2002, pp. 767--775.

\bibitem{MOO05}
E.~Mossel, R.~O'Donnell, K.~Oleszkiewicz, Noise stability of functions with low influences: invariance and optimality, in: 46th Annual IEEE Symposium on Foundations of Computer Science (FOCS'05), IEEE, 2005, pp. 21--30.

\bibitem{KKMO07}
S.~Khot, G.~Kindler, E.~Mossel, R.~O’Donnell, Optimal inapproximability results for {MAX-CUT} and other 2-variable {CSPs}?, SIAM Journal on Computing 37~(1) (2007) 319--357.

\bibitem{K99}
H.~Karloff, How good is the {Goemans--Williamson} {MAX-CUT} algorithm?, SIAM Journal on Computing 29~(1) (1999) 336--350.

\bibitem{BCIM18}
A.~E. Brouwer, S.~M. Cioabă, F.~Ihringer, M.~McGinnis, The smallest eigenvalues of hamming graphs, johnson graphs and other distance-regular graphs with classical parameters, Journal of Combinatorial Theory, Series B 133 (2018) 88--121.
\newblock \href {https://doi.org/https://doi.org/10.1016/j.jctb.2018.04.005} {\path{doi:https://doi.org/10.1016/j.jctb.2018.04.005}}.

\bibitem{janmark2014global}
J.~Janmark, D.~A. Meyer, T.~G. Wong, Global symmetry is unnecessary for fast quantum search, Physical Review Letters 112~(21) (2014) 210502.

\bibitem{TFHMG20}
R.~Tate, M.~Farhadi, C.~Herold, G.~Mohler, S.~Gupta, \href{https://doi.org/10.1145/3549554}{Bridging classical and quantum with {SDP} initialized warm-starts for {QAOA}}, ACM Transactions on Quantum Computing (jun 2022).
\newblock \href {https://doi.org/10.1145/3549554} {\path{doi:10.1145/3549554}}.
\newline\urlprefix\url{https://doi.org/10.1145/3549554}

\bibitem{TMGMG21}
R.~Tate, J.~Moondra, B.~Gard, G.~Mohler, S.~Gupta, Warm-started qaoa with custom mixers provably converges and computationally beats goemans-williamson's max-cut at low circuit depths, arXiv preprint arXiv:2112.11354 (2021).

\bibitem{egger2020warm}
D.~J. Egger, J.~Mare{\v{c}}ek, S.~Woerner, \href{https://doi.org/10.22331/q-2021-06-17-479}{Warm-starting quantum optimization}, {Quantum} 5 (2021) 479.
\newblock \href {https://doi.org/10.22331/q-2021-06-17-479} {\path{doi:10.22331/q-2021-06-17-479}}.
\newline\urlprefix\url{https://doi.org/10.22331/q-2021-06-17-479}

\bibitem{augustino2024strategies}
B.~Augustino, M.~Cain, E.~Farhi, S.~Gupta, S.~Gutmann, D.~Ranard, E.~Tang, K.~Van~Kirk, Strategies for running the qaoa at hundreds of qubits, arXiv preprint arXiv:2410.03015 (2024).

\bibitem{seidel1979strongly}
J.~Seidel, Strongly regular graphs, Surveys in combinatorics 38 (1979) 157--180.

\bibitem{bose1963strongly}
R.~C. Bose, Strongly regular graphs, partial geometries and partially balanced designs., Pacific Journal of Mathematics (1963) 389--419.

\bibitem{boyd2004convex}
S.~Boyd, S.~P. Boyd, L.~Vandenberghe, Convex optimization, Cambridge university press, 2004.

\bibitem{DGS18}
I.~Dunning, S.~Gupta, J.~Silberholz, What works best when? {A} systematic evaluation of heuristics for {Max-Cut} and {QUBO}, INFORMS Journal on Computing 30~(3) (2018).

\bibitem{haemers2001pseudo}
W.~H. Haemers, E.~Spence, The pseudo-geometric graphs for generalized quadrangles of order (3, t), European Journal of Combinatorics 22~(6) (2001) 839--845.

\bibitem{brittin1945valence}
W.~Brittin, Valence angle of the tetrahedral carbon atom, Journal of Chemical Education 22~(3) (1945) 145.

\bibitem{CI-QuBe2021}
R.~Tate, S.~Gupta, \href{https://github.com/swati1729/CI-QuBe}{{CI-QuBe}}, GitHub repository (2021).
\newline\urlprefix\url{https://github.com/swati1729/CI-QuBe}

\bibitem{WHJR18}
Z.~Wang, S.~Hadfield, Z.~Jiang, E.~G. Rieffel, Quantum approximate optimization algorithm for {MaxCut}: {A} fermionic view, Physical Review A 97~(2) (2018) 022304.

\bibitem{KMR17}
N.~Krislock, J.~Malick, F.~Roupin, Biqcrunch: A semidefinite branch-and-bound method for solving binary quadratic problems, {ACM} {Transactions} on {Mathematical} {Software} 43~(4) (Jan. 2017).

\end{thebibliography}

\appendix

\section{Other Karloff Instances}
\label{sec:other_karloff_instances}
\begin{figure}
    \centering
    \includegraphics[width=1.1\columnwidth]{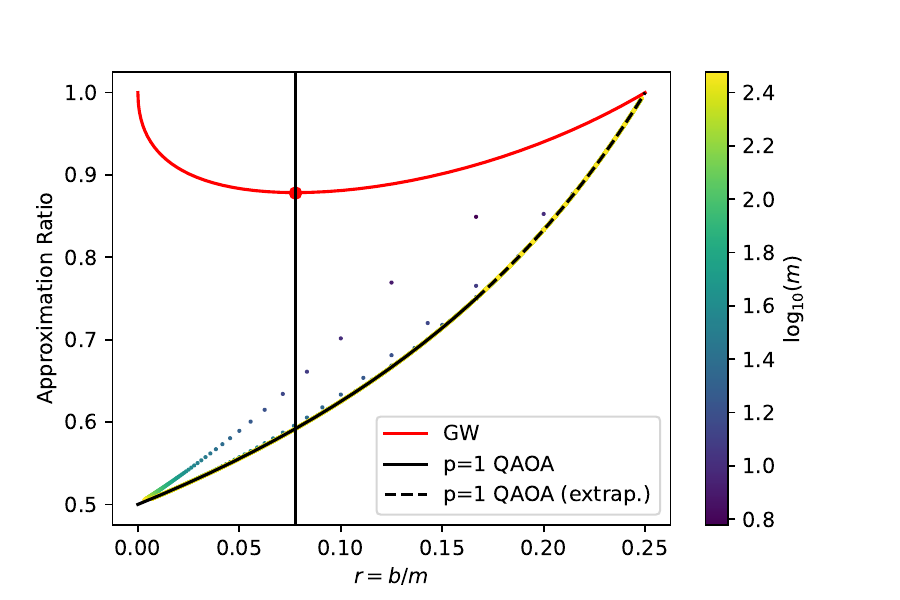}
    \caption{\footnotesize Approximation ratios for Karloff instances $J(m, m/2, b)$ for varying $r = b/m$. Red: the GW approximation ratio (which is only dependent on $r$ for any particular instance). Black: The $p=1$ QAOA limiting approximation ratio as $m \to \infty$; the solid line is given by a simple formula in $r$ for which we provably know the limiting approximation ratio, the dashed line is an extrapolation of this formula. Each dot corresponds to a specific Karloff instance, the color of which corresponds to the size of the graph (which increases with $m$). The red dot and vertical line correspond to the worst-case approximation ratio of GW (0.878) which occurs at $r\approx 0.0777$. }
    \label{fig:karloff_approx_ratios}
\end{figure}

As discussed in Section \ref{sec:GW_Alg}, the Karloff instances which approach the $0.878$-approximation ratio bound have the form $J(m, m/2, b)$ with $m$ even and $b/m$ approaching $\approx0.0777$ (corresponding to a critical angle of $\theta^* \approx 2.33$ radians between points in the optimal SDP solution). In this section, we will let $r := b/m$ denote the ratio between $b$ and $m$ for any Karloff instance.

In Section \ref{sec:qaoaPerformanceInterestingInstances}, it was observed that single-layer QAOA ($p=1)$ was unable to beat such instances (with $r \approx 0.0777$). However, it may be possible that the worst- and best-case Karloff instances for GW and QAOA are different; this warrants the study of both algorithms for different values of $r$.

In Figure \ref{fig:karloff_approx_ratios}, we consider the approximation ratios of both algorithms on Karloff instances with $0 < r < 1/4$. The $r < 1/4$ condition comes from the condition in the proven conjecture mentioned in Section \ref{sec:small_instances_karloff}; this condition was needed to prove the Max-Cut formula for Karloff instances (Equation \ref{eqn:maxcut_GW}). We exclude $r = 0$ since finding the Max-Cut of such graphs is trivial: it can be shown that such ($r=0$) Karloff instances are isomorphic to a disconnected graph where each connected component is just two vertices connected by a single edge.

\subsection{Instance-Independent Approximation Ratios}
From Theorem \ref{Thm-Karloff}, it can be shown that the approximation ratio of GW is solely a function of $r = b/m$; these GW approximation ratios are represented by the red line in Figure \ref{fig:karloff_approx_ratios}.

For $p=1$ QAOA, the approximation ratio is \emph{not} solely a function of $r=b/m$; however, for any choice of fixed $r < 1/6$, the approximation ratio can be calculated in the limit as $m \to \infty$. In particular, slight algebraic modifications to the equations in the proof of Theorem \ref{thm:karloffQAOAApproxRatio} give that this limiting approximation ratio (as $m \to \infty$) is given by:
\begin{equation}\label{eqn:limiting_approx_ratio_karloff}
\alpha_\text{QAOA} = \frac{1}{2-4r}.
\end{equation}

The condition $r < 1/6$ above comes from the fact that Theorem \ref{thm:karloffQAOAApproxRatio} utilizes that Karloff instances are triangle-free for $r < 1/6$ (see proof of Theorem \ref{thm:karloffTriangleFree}).

Equation \ref{eqn:limiting_approx_ratio_karloff} is plotted in black in Figure \ref{fig:karloff_approx_ratios}; the solid part ($0 < r < 1/6$) indicates where the formula provably holds and the dashed part is an extrapolation of the formula for $1/6 \leq r < 1/4$. As we will soon see, empirical evidence suggests that this extrapolation is correct, i.e., that Equation \ref{eqn:limiting_approx_ratio_karloff} corresponds to the limiting QAOA approximation ratio for all $0 < r < 1/4$.

\subsection{Instance-Specific Approximation Ratios}
Figure \ref{fig:karloff_approx_ratios} also shows the instance-specific approximation ratios of specific Karloff instances for all karloff instances $J(m, m/2, b)$ with even $m$ up to $m=300$ and all integers $b$ with $0 < b < m/4$. This amounts to 5550 different Karloff instances whose $r$-values are roughly\footnote{Due to the way the set of instances is constructed, there are some spikes in frequency for $r$-values that correspond to simplified fractions with small denominators (e.g. $r = 1/5, 1/6, 1/8,$ etc).} distributed uniformly in the interval $r \in (0,1/4)$. These approximation ratios were calculated using the grid search described in Section \ref{sec:qaoaPerformanceInterestingInstances}; though with a smaller grid ($1000 \times 1000$) due to computational limitations. We remark that in the case that $r < 1/6$, the instances are triangle-free and thus no grid search is needed since the optimal parameters are known in this case (see proof of Theorem \ref{thm:karloffQAOAApproxRatio}).

For large values of $m$ (corresponding to the yellow dots of Figure \ref{fig:karloff_approx_ratios}), we see that the instance-specific approximation ratios indeed approach the limiting approximation ratio given by Equation \ref{eqn:limiting_approx_ratio_karloff} (the black lines), even in the case where we have extrapolated the formula to $1/6 \leq r < 1/4$ (dashed line).

\subsection{Approximation as $r\to 0$}
Lastly, we remark on the approximation ratios of both algorithms at the extremes: $r=0$ and $r=1/4$. As previously mentioned, at $r=0$, the Karloff instances are just a collection of disjoint edges. It is straightforward to show that both GW and $p=1$ QAOA achieve approximation ratios of 1.0 (due to the fact that each algorithm can obtain the maximum cut for a single-edge graph). Interestingly, the approximation ratio for QAOA is not continuous at $r=0$: while the approximation is 1.0 at $r=0$, it approaches an approximation ratio of:
$$\lim_{r \to 0^+} \frac{1}{2-4r} = \frac{1}{2},$$
as given by Equation \ref{eqn:limiting_approx_ratio_karloff} and as seen in Figure \ref{fig:karloff_approx_ratios}.

We also remark that for any sequence of Karloff graphs whose $r$-values are non-zero but approach 0, it must necessarily be the case that the sequence of $m$-values for such graphs tends to infinity; this explains the convergence of the instance-specific approximation ratios (the individual dots of Figure \ref{fig:karloff_approx_ratios}) to $1/2$ as $r \to 0^+$.

\subsection{Approximation as $r\to 1/4$}
On the other extreme, remarkably, both algorithms appear to approach an approximation ratio of $1$ as $r$ approaches $1/4$ from the left, regardless of the size of the graph. This occurs because as $r \to 1/4$, one can show that the Max-Cut value approaches half of the edges. Let $G_1, G_2, \dots$ be a sequence of Karloff graphs $J(m, m/2, b)$ whose corresponding $r$-values ($r_1, r_2, \dots)$ approach $1/4$, then:
\begin{align*}
    &=\lim_{j \to \infty} \frac{\mc(G_j)}{|E_j|}\\
    &= \lim_{j \to \infty} \frac{\frac{n_j}{2}{m_j/2 \choose b_j}^2(1 - 2r_j)}{|E_j|} \tag{see Theorem \ref{thm:karloffQAOAApproxRatio} proof}\\
    &= \lim_{j \to \infty}\frac{ \frac{n_j}{2}{m_j/2 \choose b_j}^2(1 - 1/2)}{|E_j|} \tag{as $r_j \to 1/4$}\\
   &= \lim_{j \to \infty}\frac{\frac{1}{2}\frac{n_jd_j}{2}}{|E_j|} \tag{Lemma \ref{thm:karloffDegree}} \\
   &= \lim_{j \to \infty}\frac{\frac{1}{2}|E_j|}{|E_j|} \tag{handshaking lemma} \\
   &= \frac{1}{2}.
\end{align*}

For $p=1$ QAOA, at $\gamma=0$ and $\beta = 0$, the algorithm will cut (in expectation) exactly half of the edges, and hence, if the Max-Cut approaches half of the edges (which we've shown is the case as $r \to 1/4$), then the approximation ratio will approach 1.

For the GW algorithm, as $r \to 1/4$, the corresponding angle $\theta$ between points in the SDP approaches $\theta = \arccos(4r-1) \to  \pi/2$, meaning each edge has a probability approaching $\frac{\theta}{\pi} \to \frac{\pi/2}{\pi} = \frac{1}{2}$ of being cut, and thus half of the edges are cut in expectation, giving an approximation ratio that approaches 1 as the Max-Cut itself also approaches half the edges as previously shown.

\section{Mixed Weight Instances with Large Magnitude Range}
\label{sec:mixedWeightInstances}

\begin{figure}[t]
    \centering
    \includegraphics[scale=0.6]{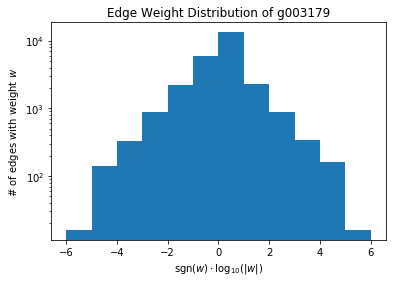}
    \caption{\footnotesize The edge-weight distribution of a challenge instance ``g003179" for classical heuristics, with both positive and negative edge weights spanning several orders of magnitude; 9 other challenging MQLib instances had similar distributions (see Table \ref{tab:MQLIB_Instances}, all except g000435 and g001349.}
    \label{fig:edgeWeights}
\end{figure}

\begin{table*}[!t]
\centering
\def\arraystretch{1.3}
\begin{tabular}{|c|c|c|c|c|c|c|c|}
\hline
     Instance &  Nodes &  Edges & A.R. & Expected GW Ratio & GW Time & Neg. Weights & Magnitude Range \\\hline
     g000330 & 300 & 14011 &  176  & 0.9814 & 359.34525442123413 & Yes & 5.308\\
     g000417 & 300 & 34753 & 176  & 0.9699 & 321.5829267501831& Yes & 5.388\\
     g000435 & 420 & 619 & 247  & 0.9731 & 880.8003544807434 & No & 2.581\\
     g000572 & 200 & 18246 & 120 & 0.9675 & 250.95755791664124& Yes & 5.365\\
     g000723 & 200 & 18230 & 120  & 0.9532 & 130.70434188842773& Yes & 5.316\\
     g000762 & 250 & 26522 & 147  & 0.9642 & 409.95273995399475& Yes & 5.443\\
     g001349 & 640 & 960 & 377 & 0.9619 & 838.4398355484009 & No & 0.566\\
     g001361 & 250 & 12031 & 147 & 0.9864 & 192.36800813674927& Yes & 5.267\\
     g002581 & 200 & 18237 & 120 & 0.9695 & 249.02441000938416& Yes & 5.380\\
     g002935 & 300 & 13983 & 176 & 0.9789 & 235.74102520942688& Yes & 5.342\\
     g003179 & 250 & 26534 & 147 & 0.9689 & 169.3680293560028& Yes & 5.263\\
    \hline
\end{tabular} 
\caption{\label{tab:MQLIB_Instances}\footnotesize The table includes all the instances in the MQLIB library where no MQLIB heuristic was able to obtain 99.9\% of the optimal solution within 5\% of the allotted runtime for the instance. The columns correspond to the instance name, the number of nodes, nonzero-weight edges, the allotted runtime (A.R.), the GW ratio (the SDP solution was obtained via a solver with the expectation then being calculated analytically), the time needed to run the GW solver, whether the instance has negative-weighted edges, and the range of edge-weight magnitudes (defined as $\max_{e \in E} \log_{10}(|w_e|) - \min_{e \in E} \log_{10}(|w_e|)$ across all edges with non-zero edge weights).}
\end{table*}
Max-Cut instances that are difficult across a large collection of heuristics may also be of practical interest. To help us find such instances, we analyzed the data collected by Dunning et al.'s large-scale experimental study consisting of 38 Max-Cut heuristics and 3,296 instances in \href{https://github.com/MQLib/MQLib}{MQLib} \cite{DGS18}. In their work, they used machine learning in order to be able to predict which heuristic would perform best on which instances depending on certain graph properties. We used the semidefinite-based exact solver, \href{https://web.archive.org/web/20250113192126/https://biqcrunch.lipn.univ-paris13.fr/}{BiqCrunch}, to find the optimal solution for as many instances as possible \cite{KMR17}. Next, we asked if the rate of closing the gap to Max-Cut was slower on some instances. In particular, we asked if there existed some instances where none of the implemented heuristics found a cut within 99.9\% of the Max-Cut within 5\% of the allotted instance-specific run time.  We found that there were only 11 instances, with less than 700 nodes, where this happened, with a specific edge-weight distribution. In particular, 9 out of 11 of these graphs had a number of nodes between 200 and 300, and had both positive and negative edge weights spanning several orders of magnitude, as seen in Figure \ref{fig:edgeWeights}. See Table \ref{tab:MQLIB_Instances} for details.

\end{document}